\documentclass{article}
\usepackage[a4paper, total={6in, 8in}]{geometry}
\usepackage{mathrsfs}
\usepackage{verbatim}
\usepackage{amsmath}
\usepackage{amsthm}
\usepackage{amssymb}
\usepackage{bbold}
\usepackage{bm}
\usepackage{multirow}
\usepackage{fixltx2e}
\usepackage{graphicx}
\usepackage{float}
\usepackage{comment}
\usepackage{biblatex}
\usepackage{amsfonts}
\usepackage{authblk}


\newcounter{parentnumber}
\numberwithin{equation}{section}

\theoremstyle{plain}
\newtheorem{Assump}{Assumption}[section]
\newtheorem{Cond}{Condition}[section]
\newtheorem{Theo}{Theorem}[section]

\newtheorem{Lemma}{Lemma}[section]

\newtheorem{Algor}{Algorithm}[section]

\theoremstyle{remark}

\newtheorem{Remark}{Remark}[section]

\newcommand{\ep}{\epsilon}
\newcommand{\vep}{\varepsilon}

\newcommand{\sig}{\sigma}
\newcommand{\Sig}{\Sigma}
\newcommand{\sumn}{\sum_{t=1}^{n}}

\newcommand{\Ps}{P^{*}}
\newcommand{\Es}{E^{*}}
\newcommand{\Vars}{Var^{*}}
\newcommand{\Covs}{Cov^{*}}
\newcommand{\hphi}{\hat{\phi}}
\newcommand{\hphis}{\hat{\phi}^{*}}

\newcommand{\ts}{t^{*}}

\newcommand{\hV}{\hat{V}}
\newcommand{\hVs}{\hat{V}^{*}}
\newcommand{\bV}{\bm{V}}
\newcommand{\cV}{\check{V}}
\newcommand{\Vs}{V^{*}}
\newcommand{\bVs}{\bm{V}^{*}}
\newcommand{\hbV}{\hat{\bm{V}}}
\newcommand{\cbV}{\check{\bm{V}}}
\newcommand{\bep}{\bm{\ep}}
\newcommand{\be}{\bm{e}}
\newcommand{\bvep}{\bm{\vep}}
\newcommand{\hgamma}{\hat{\gamma}}
\newcommand{\hrho}{\hat{\rho}}
\newcommand{\hep}{\hat{\ep}}
\newcommand{\he}{\hat{e}}
\newcommand{\hl}{\hat{l}}
\newcommand{\hbep}{\hat{\bm{\ep}}}
\newcommand{\hbe}{\hat{\bm{e}}}
\newcommand{\cep}{\check{\ep}}
\newcommand{\ce}{\check{e}}
\newcommand{\cvep}{\check{\vep}}
\newcommand{\cbep}{\check{\bm{\ep}}}
\newcommand{\cbe}{\check{\bm{e}}}
\newcommand{\cbvep}{\check{\bm{\vep}}}
\newcommand{\hsig}{\hat{\sig}}
\newcommand{\hSig}{\hat{\Sig}}
\newcommand{\tSig}{\tilde{\Sig}}
\newcommand{\beps}{\bm{\ep}^{*}}
\newcommand{\bes}{\bm{e}^{*}}
\newcommand{\bveps}{\bm{\vep}^{*}}
\newcommand{\bxis}{\bm{\xi}^{*}}
\newcommand{\lnrl}{\lfloor nr \rfloor}
\newcommand{\lnsl}{\lfloor ns \rfloor}
\newcommand{\lnul}{\lfloor nu \rfloor}
\newcommand{\lnvl}{\lfloor nv \rfloor}
\newcommand{\lnwl}{\lfloor nw \rfloor}
\newcommand{\sigs}{\sig^{*}}
\newcommand{\Ys}{Y^{*}}
\newcommand{\Op}{O_{p}(1)}
\newcommand{\op}{o_{p}(1)}

\title{Linear Process Bootstrap Unit Root Test}
\author{Nan Zou\thanks{Email address: nzou@ucsd.edu.} }
\author{Dimitris Politis\thanks{Email address: dpolitis@ucsd.edu.}}
\affil{Department of Mathematics, University of California-San Diego, \newline La Jolla, CA 92093}
\date{\vspace{-5ex}}
\begin{document}
\maketitle

\begin{abstract}
One of the most widely applied unit root test, Phillips-Perron test, enjoys in general high powers, but suffers from size distortions when moving average noise exists. As a remedy, this paper proposes a nonparametric bootstrap unit root test that specifically targets moving average noise. Via a bootstrap functional central limit theorem, the consistency of this bootstrap approach is established under general assumptions which allows a large family of non-linear time series. In simulation, this bootstrap test alleviates the size distortions of the Phillips-Perron test while preserving its high powers.     
\end{abstract}

\section{Introduction}
Among extensive literature on unit root test, the Augmented Dickey-Fuller (ADF) test and the Phillips-Perron (PP) test are perhaps the most renowned. When put into simulation, PP test has been found to enjoy higher power than ADF test but suffers greater size distortion, especially under negative Moving Average (MA) noises (\textcite{phillips1988,nabeya1994,cheung1997,leybourne1999}). For a solution to this size distortion, see \textcite{perron1996}.\\ 
\\
We propose a bootstrap unit root test as a remedy. When the asymptotic distributions of the test statistics involve unknown parameters, bootstrap circumvents the estimation of the unknown parameters and as a result eases the hypothesis test. On the other hand, when the asymptotic distributions are pivotal, bootstrap unit root test may enjoy second order efficiency, and may consequently reduce the aforementioned size distortion (\textcite{park2003}). Variants of bootstrap unit root test include AutoRegressive (AR) sieve bootstrap test (\textcite{psaradakis2001,palm2008}), block bootstrap test (\textcite{paparoditis2003}),  stationary bootstrap test (\textcite{swensen2003,parker2006}), and wild bootstrap test (\textcite{cavaliere2009}). \\
\\
To target the size distortion with MA noise, we apply Linear Process Bootstrap (LPB) of \textcite{mcmurry2010} to unit root test. As the closest analogue of MA-sieve bootstrap, LPB first estimates the autocovariance matrix of the noise, then pre-whitens the noise with the estimated autocovariance matrix, then bootstraps from the pre-whitened noise, and finally post-colors the bootstrap noise with the the estimated autocovariance matrix. In sample mean case, \textcite{mcmurry2010, jentsch2015} indicate good asymptotic and empirical performance of LPB, particularly in the presence of MA noise. \\
\\
As a result, LPB unit root test becomes a promising solution to the size distortion under MA noise. To develop a large sample theory for LPB unit root test, we extend the bootstrap Central Limit Theorem (CLT) for LPB into a regression setting, establish a bootstrap Functional CLT (FCLT) for LPB, and prove the consistency of this bootstrap method. Despite its name, LPB unit root test turns out to be asymptotically valid under not only linear noises but also a large family of non-linear noises, i.e., the physical dependent process defined in \textcite{wu2005}. \\
%The asymptotic validity of the adaptive bandwidth selection of \textcite{politis2003} and prove it asymptotic validity in unit root test. \\
\\
This paper proceeds as follows. Section \ref{sec:PP} specifies the physical dependence assumption and recalls the popular Phillips-Perron test. Section \ref{sec:LPB} introduces LPB unit root test, details the estimation of the autocovariance matrix, and describes the adaptive bandwidth selection. Section \ref{sec:simulation} presents the empirical results of LPB unit root test. Appendix includes all technical proofs. 
\section{Phillips-Perron Test}\label{sec:PP}
Suppose $\{Y_{t}\}_{t=1}^{n}$ is observable. For $t\in\mathbb{N}^{+} $, define $\phi_{t}$ and $V_{t}$ as the prediction coefficient and error, respectively, when predicting $Y_{t}$ with $Y_{t-1}$. Suppose $\phi_{t}=\phi$ for all $t\in \mathbb{N}^{+}$. Then
\begin{equation}\label{eqn:model}
Y_{t}=\phi Y_{t-1}+V_{t}. 
\end{equation}
Now we assume the noise sequence $\{V_{t}\}_{t\in \mathbb{Z}}$ is strictly stationary, short-range dependent, and invertible. Specifically, consider the following assumptions on $\{V_{t}\}_{t\in \mathbb{Z}}$.
\begin{Assump}\label{assump:weak}
Let $\{\ep_{t}\}_{t\in \mathbb{Z}}$ be a sequence of i.i.d. random variables. Let $\ep_{0}'$ be identically distributed with $\ep_{0}$, and be independent of $\{\ep_{t}\}_{t\in \mathbb{Z}}$. Suppose $V_{t}=g(...,\ep_{t-1},\ep_{t})$. Let $V_{t}'=g(...,\ep_{-1},\ep_{0}',\ep_{1},...,\ep_{t})$, and let $\delta_{p}(t)=(E(|V_{t}-V_{t}'|^{p}))^{1/p}$ be the physical dependence measure of $\{V_{t}\}$. Suppose $\sum_{t=1}^{\infty}\delta_{4}(t)<\infty$. Let $\gamma(h)=E(V_{t}V_{t-h})$. Suppose $\sum_{h\in \mathbb{Z}}\gamma(h)>0$, $\sum_{h=0}^{\infty}h|\gamma(h)|<\infty$, $E(V_{t})=0$, and $E(V_{t}^{4})<\infty$.
\end{Assump}
\begin{Assump}\label{assump:strong}
Recall Assumption \ref{assump:weak}. Further assume that for some $p>4$, $\sum_{t=1}^{\infty}\delta_{p}(t)<\infty$ and $E(|V_{t}|^{p})<\infty$; for some $\beta>2$, $|\gamma(h)|=o(h^{-\beta})$; for some $\alpha>0$,  $h^{\alpha}\sum_{k=h+1}^{\infty}|\gamma(k)|$ is non-increasing when $h$ is large enough.
\end{Assump}
When $\phi=1$, suppose $Y_{0}=0$; then $\{Y_{t}\}_{t\in \mathbb{N}^{+}}$ is a unit root process starting at zero. When $\phi<0$, suppose \eqref{eqn:model} holds for all $t\in \mathbb{Z}$; then $\{Y_{t}\}_{t\in \mathbb{N}^{+}}$ is a strictly stationary process. To separate these two cases, we test 
\begin{equation}
H_{0}: \ \phi=1 \ \text{vs}\ H_{1}: \phi<1. 
\end{equation}
The famous PP test centers on the OLS estimator $\hphi$ in $Y_{t}=\hphi Y_{t-1}+\hV_{t}$, and its t-statistic $t$. Under Assumption \ref{assump:weak}, the asymptotic null distributions of $\hphi$ and $t$ results from the FCLT in Lemma \ref{lemma:FCLT}. 
\begin{Lemma}[\textcite{wu2005}]\label{lemma:FCLT}
Suppose Assumption \ref{assump:weak} holds. Let $\sig^{2}=Var(n^{-1/2} \sumn V_{t})$, $S(u)=n^{-1/2}\sig^{-1} \sum_{t=1}^{\lnul} V_{t}$. Let $W(u)$ be a standard Brownian motion. If $\phi=1$, $S\Rightarrow W$. 
\end{Lemma} 
%Notice that we include no deterministic trend in \eqref{eqn:model}
\section{Linear Process Bootstrap Unit Root Test}\label{sec:LPB}
As mentioned in introduction, PP test enjoys high empirical powers, but suffers from empirical size distortions under negative MA noise. To mitigate the size distortion while preserving the high power, we introduce LPB unit root test below. The name of LPB follows from the fact that the bootstrapped noise is a linear process.\\
\\
%In this way, the bootstrap noise created by the LPB preserve the second moment properties of the original noise. Since the asymptotic  distributions of the PP test statistics under the null only contain second moment of the noise, LPB unit root test is expected to be valid.\\
Let $\bV=(V_{1},...,V_{n})'$, $\hbV=(\hV_{1},...,\hV_{n})'$, $\bar{\hbV}=(\bar{\hV}_{1},...,\bar{\hV}_{n})'$,
$\bVs=(\Vs_{1},...,\Vs_{n})'$,
$\hbep=(\hep_{1},...,\hep_{n})'$, and $\beps=(\ep_{1}^{*},...,\ep_{n}^{*})'$. Let  $\Sig=Var(\bV)$ and $\hSig_{\hV}$ be a positive definite estimator of $\Sig$. In Algorithm \ref{algor:matrix} we will further specify $\hSig_{\hV}$. Let ${\hSig_{\hV}}^{1/2}$ be a lower triangular matrix that satisfies Cholesky decomposition  ${\hSig_{\hV}}^{1/2}{\hSig_{\hV}}^{{1/2}'}=\hSig_{\hV}$, and ${\hSig_{\hV}}^{-1/2}$ be the inverse matrix of ${\hSig_{\hV}}^{1/2}$. Let $\bar{Y}_{t}=n^{-1}\sumn{Y}_{t}$, $\bar{V}_{t}=n^{-1}\sumn{V}_{t}$, $\bar{\hV}_{t}=n^{-1}\sumn{\hV}_{t}$, $\bar{\hep}_{t}=n^{-1}\sumn{\hep}_{t}$, and $\hsig_{\hep}^{2}=n^{-1}\sumn(\hep_{t}-\bar{\hep}_{t})^{2}$. Let $\Ps$, $\Es$, $\Vars$, $\Covs$ be the probability, expectation, variance, and covariance, respectively, conditional on data $\{Y_{t}\}$.
\begin{Algor}\emph{[Linear process bootstrap unit root test]}\label{algor:LPB}
\item 
Step 1: regress $Y_{t}=\hphi Y_{t-1}+\hV_{t}$; record $\hphi$ and its t-statistic $t$. \\
Step 2: let $\cV_{t}=\hV_{t}-\bar{\hV}_{t}$; let $\hbep=\hSig_{\hV}^{-1/2}\cbV$; let $\cep_{t}=(\hep_{t}-\bar{\hep}_{t})/\hsig_{\hep}$. \\
Step 3: randomly sample $\ep_{1}^{*},...,\ep_{n}^{*}$ from  $\{\cep_{1},...,\cep_{n}\}$. \\
Step 4: let $\bVs=\hSig_{\hV}^{1/2}\beps$; let $\Ys_{t}=\Ys_{t-1}+\Vs_{t}$ and $\Ys_{0}=0$. \\
Step 5: regress $\Ys_{t}=\hphis \Ys_{t-1}+\hVs_{t}$; record $\hphis$ and its t-statistic $\ts$.\\
Step 6: run Step 3-5 for $B$ times and get $\{\hphis_{1},...,\hphis_{B}\}$ and $\{\ts_{1},...,\ts_{B}\}$. \\
Step 7: reject the null if $B^{-1}\sum_{i=1}^{B}1\{\hphi>\hphis_{i}\}<size$, or alternatively, $B^{-1}\sum_{i=1}^{B}1\{t>\ts_{i}\}<size$. 
\end{Algor}

Now we specify $\hSig_{\hV}$, the estimator of the autocovariance matrix $\Sig$. Noticing the inconsistency of the sample autocovariance matrix, \textcite{mcmurry2010} propose a new autocovariance matrix estimator $\hSig_{\hV}$ detailed in Algorithm \ref{algor:matrix} below. By construction, $\hSig_{\hV}$ is positive definite and possesses a banded structure. By letting the bandwidth of this banded structure goes to infinity as sample size goes to infinity, $\hSig_{\hV}$ becomes a consistent estimator of $\Sig$. Since the autocovariance matrix of a finite-order MA process has as well a banded nature, $\hSig_{\hV}$ constitutes a MA-sieve estimator of $\Sig$ and hence performs especially well with MA noise. 
%see \textcite{mcmurry2010}. 
\begin{Algor}[\textcite{mcmurry2010}]\emph{[Estimation of the autocovariance matrix]}\label{algor:matrix}
\item
Let $\hgamma_{\hV}(h)=n^{-1}\sum_{t=|h|+1}^{n}\hV_{t}\hV_{t-|h|}$. Define kernel function $\kappa(\cdot)$ by 
$$\kappa(x)=\begin{cases}
1,      & \text{if } |x| \leq 1, \\
g(|x|),  & \text{if } 1<|x| \leq c_{\kappa}, \\
0,      & \text{if } |x| > c_{\kappa}. \\
\end{cases}$$ 
where $g(\cdot)$ is a function satisfying $|g(x)|<1$, and $c_{\kappa}$ is a constant satisfying $c_{\kappa}\geq 1$. An example of kernel function $\kappa(x)$ is the trapezoid kernel of \textcite{politis1995}:
\begin{equation}\label{eqn:kernel}
\kappa(x)=\begin{cases}
1,      & \text{if } |x| \leq 1, \\
2-|x|,  & \text{if } 1<|x| \leq 2, \\
0,      & \text{if } |x| > 2. \\
\end{cases}
\end{equation}
Let $\kappa_{l}=\kappa(x/l)$, where $l$ is a kernel bandwidth to be determined. Define tapered covariance matrix estimator $$\tSig_{\hV}=[\kappa_{l}(i-j)\hgamma_{\hV}(i-j)]_{i,j=1}^{n}.$$ Suppose $\tSig_{\hV}=TDT',$ where $T$ is orthogonal and $D=diag(d_{1},...,d_{n})$ is diagonal. Let $\hat{d}_{j}=\max(d_{j},\hgamma_{0}n^{-1})$, $\hat{D}=diag(\hat{d}_{1},...,\hat{d}_{n})$, and $\hSig_{\hV}=T\hat{D}T'$. 
\end{Algor}
Under certain conditions on the kernel bandwidth $l$, Algorithm \ref{algor:LPB} and Algorithm \ref{algor:matrix} together present a consistent bootstrap approach. First, a bootstrap FCLT with respect to bootstrap noise $\{\Vs_{t}\}$ is established in Lemma \ref{lemma:FCLT}. Based on this bootstrap FCLT, the conditional distributions of the bootstrap statistics $\hphis$ and $\ts$ converge to the asymptotic null distributions of $\hphi$ and $t$, respectively. Hence the justification of LPB unit root test, and as a byproduct the validity of LPB in regression. 
\begin{Cond}\label{cond:bandwidth}
Let $r_{n}=r_{n}(l)=ln^{-1/2}+\sum_{h=l}^{\infty}|\gamma(h)|$. Suppose $l=l_{n}$ satisfies $r_n=O(n^{-1/4}).$
\end{Cond}
\begin{Remark}
There exists $l=l_{n}$ such that Condition \ref{cond:bandwidth} holds, for example, $l=n^{1/4}$. Together with Assumption \ref{assump:weak}, Condition \ref{cond:bandwidth} guarantees that the operator norm of $\hSig_{\hV}-\Sig$ decays at a rate faster or equal to $n^{-1/4}$, and as a result the measures of the partial sum processes $\{S^{*}(u)\}$ in Lemma \ref{lemma:bootstrap FCLT} are tight.
\end{Remark}
\begin{Lemma}\label{lemma:bootstrap FCLT}
Suppose Assumption \ref{assump:weak} and Condition \ref{cond:bandwidth} hold. Let ${\sigs}^{2}=\Vars(n^{-1/2} \sum_{t=1}^{n}\Vs_{t})$, $S^{*}(u)=n^{-1/2}{\sigs} ^{-1} \sum_{t=1}^{\lnul}\Vs_{t}$. Let $W(u)$ be a standard Brownian motion. Then no matter if $\phi=1$ or $\phi<1$, $S^{*}\Rightarrow W$ in probability.
\end{Lemma}
\begin{Theo} \label{theo:consistency}
Suppose Assumption \ref{assump:weak} and Condition \ref{cond:bandwidth} hold. Let $P^{H_{0}}$ be the probability measure corresponding to the null hypothesis. Then 
$$\sup_{x}|\Ps(n(\hphis-1)\leq x)-P^{H_{0}}(n(\hphi-1)\leq x)|=\op,$$
$$\sup_{x}|\Ps(\ts\leq x)-P^{H_{0}}(t\leq x)|=\op.$$
\end{Theo}
To implement Algorithm \ref{algor:matrix}, we choose bandwidth $l$ according to the adaptive bandwidth selection of \textcite{politis2003} in Algorithm \ref{algor:bandwidth}. Lemma \ref{lemma:bandwidth} shows the bandwidth selected by Algorithm \ref{algor:bandwidth} satisfies Condition \ref{cond:bandwidth}. The validity of the bandwidth selection method follows immediately in Theorem \ref{theo:bandwidth}. Notice that when validating the bandwidth selection method, \textcite{politis2003} and \textcite{mcmurry2010} require the autocovariance function $\gamma(h)$ to be either polynomial, exponential, or truncated. In contrary, our assumptions in Theorem \ref{theo:bandwidth} are much more general.
\begin{Algor}[\textcite{politis2003}]\emph{[Selection of the bandwidth]}\label{algor:bandwidth}\item
Let $\hrho_{\hV}(h)=\hgamma_{\hV}(h)/\hgamma_{\hV}(0)$. Select bandwidth $\hl$ as the smallest positive integer satisfying 
$$|\hrho_{\hV}(\hl+k)|<c(\log n)^{1/2}n^{-1/2}, \ k=1,...,K_{n},$$
where $K_{n}$ is a positive, non-decreasing sequence such that $K_{n}=o(\log n)$, and $c$ is a positive constant. 
\end{Algor}

\begin{Lemma}\label{lemma:bandwidth}
Select bandwidth $\hl$ by Algorithm \ref{algor:bandwidth}. Under Assumption \ref{assump:weak} and \ref{assump:strong},
\begin{equation*}
    \hl n^{-1/2}+\sum_{h=\hl+1}^{\infty}|\gamma(h)|=O_{p}(n^{-1/4}).
\end{equation*}
\end{Lemma}
\begin{Theo}\label{theo:bandwidth}
Select bandwidth $\hl$ by Algorithm \ref{algor:bandwidth}. Then under Assumption \ref{assump:weak} and \ref{assump:strong}, the results in Theorem \ref{theo:consistency} hold. 
\end{Theo}
\section{Simulation}\label{sec:simulation}
\subsection{Data Generating Process}

Let $X_{t}-X_{t-1}=\varphi X_{t-1}+V_{t}$. The values of $\varphi$ are set to be $0,-0.02,-0.04,-0.06,-0.08,-0.10$ in order to generate the power curve. Let $\{V_{t}\}$ be generated by Table \ref{table:DGP} below, and $\{\ep_{t}\}\sim iid \ N(0,1)$. 
% Please add the following required packages to your document preamble:
% 
\begin{table}[H]
\centering
\caption{Types of Noises}
\label{table:DGP}
\begin{tabular}{|c|c|c|}
\hline
\multirow{6}{*}{\begin{tabular}[c]{@{}c@{}}Noises\end{tabular}} & iid                     & $V_{t}=\epsilon_{t}$                                                                                           \\ \cline{2-3} 
                                                                      & ma\textsubscript{pos} & $V_{t}=\epsilon_{t}+0.5\epsilon_{t-1}$                                                                         \\ \cline{2-3} 
                                                                      & ma\textsubscript{neg} & $V_{t}=\epsilon_{t}-0.5\epsilon_{t-1}$                                                                         \\ \cline{2-3} 
                                                                      & ar\textsubscript{pos} & $V_{t}=\epsilon_{t}+0.5V_{t-1}$                                                                                \\ \cline{2-3} 
                                                                      & ar\textsubscript{neg} & $V_{t}=\epsilon_{t}-0.5V_{t-1}$                                                                                \\ \cline{2-3} 
                                                                      & arch                    & \begin{tabular}[c]{@{}c@{}}$V_{t}=\sigma_{t}\ep_{t}$, \\ $\sigma_{t}^{2}=10^{-6}+0.25V_{t-1}^{2}$\end{tabular} \\ \hline
\end{tabular}
\end{table}

\subsection{Methods}
In Table \ref{table:method} we list the unit root tests we include in the simulation. In ADF and ARB-ADF we select lag order by Modified Akaike Information Criterion (MAIC) of \textcite{ng2001}. In FPP we harness the flat-top kernel spectral density estimator of \textcite{politis1995}, and choose the kernel bandwidth according to the adaptive bandwidth selection of \textcite{politis2003}. The validity of FPP under Assumption \ref{assump:weak} results from \textcite{shao2007}. In LPB-PP we harness the trapezoid kernel stated in \eqref{eqn:kernel}. In CBB-PP we apply circular block bootstrap of \textcite{politis1992}.  The block size of CBB-PP comes from the automatic block-length selection of \textcite{politis2004}. We conjecture the validity of CBB-PP on the basis of the validity of the block bootstrap PP test. The nominal sizes of all tests are set to be 0.05. Each sample has length 100. In bootstrap methods, 500 bootstrap replicates are generated. To estimate the powers of the tests, 600 tests are conducted. Tests based on both $\hphi$ and its t-statistic $t$ are simulated. Unpublished simulation shows for each of the test listed in Table \ref{table:method}, the version based on $\hphi$ are inferior to the version based on the t-statistic $t$. Therefore, we only report the results of the tests based on the t-statistic $t$. 

\begin{table}[H]
\centering
\caption{Type of Tests}
\label{table:method}
\begin{tabular}{|l|l|}
\hline
ADF     & ADF test                         \\ \hline
ARB-ADF & AR-sieve Bootstrap ADF test      \\ \hline
FPP     & Flat-top pivoted PP test         \\ \hline
LPB-PP  & Linear Process Bootstrap PP test \\ \hline
CBB-PP  & Circular Block Bootstrap PP test \\ \hline
\end{tabular}
\end{table}
\subsection{Results}
The results in Table \ref{table:result} separate the tests into two categories. ADF and ARB-ADF, as parametric tests, show better empirical sizes, particularly under negative moving average noise. On the other hand, the nonparametric tests, i.e., FPP, LPB-PP, and CBB-PP, attain higher powers. These high powers of the nonparametric tests not only stand out under conditional heteroscedastic noise, but also occur in other cases, e.g., when positive moving average noise occurs. \\
\\
Now we focus on nonparametric tests, i.e., FPP, LPB-PP, and CBB-PP. First,  recall that FPP estimates the spectral density with flat-top kernel and adaptive bandwidth selection. We found this estimation leads to much better empirical size, compared to other popular kernel-based spectral density estimations. See \textcite{kim1990, perron1996} for  evidences under the same or similar simulation settings. \\
\\
Second, among FPP, LPB-PP, and CBB-PP, our LPB-PP achieves the best overall performance in empirical sizes and powers. More specifically, while these three tests have almost equally high powers, the LPB-PP distorts the size less under two of the least favorable noises, i.e., negative moving average noise and positive moving average noise. However, LPB-PP does not fully eradicate the size distortion problem. 

\begin{table}[H]
\centering
\caption{Sizes and (unadjusted) Powers}
\label{table:result}
\begin{tabular}{|c|c|c|c|c|c|c|c|}
\hline
                                & $\varphi$  & iid   & ma\textsubscript{pos} & ma\textsubscript{neg} & ar\textsubscript{pos} & ar\textsubscript{neg} & arch \\ \hline
\multirow{6}{*}{ADF}   & 0.00  & 0.060 & 0.048      & 0.088      & 0.048      & 0.047      & 0.020   \\ \cline{2-8} 
                         & -0.02 & 0.150 & 0.133      & 0.207      & 0.140      & 0.130      & 0.052   \\ \cline{2-8} 
                         & -0.04 & 0.292 & 0.275      & 0.300      & 0.273      & 0.273      & 0.098   \\ \cline{2-8} 
                         & -0.06 & 0.448 & 0.398      & 0.415      & 0.425      & 0.415      & 0.168   \\ \cline{2-8} 
                         & -0.08 & 0.585 & 0.517      & 0.515      & 0.518      & 0.523      & 0.230   \\ \cline{2-8} 
                         & -0.10 & 0.733 & 0.600      & 0.593      & 0.642      & 0.610      & 0.333   \\ \hline
\multirow{6}{*}{ARB-ADF} & 0.00  & 0.057 & 0.050      & 0.065      & 0.043      & 0.050      & 0.045   \\ \cline{2-8} 
                         & -0.02 & 0.153 & 0.157      & 0.182      & 0.158      & 0.118      & 0.127   \\ \cline{2-8} 
                         & -0.04 & 0.248 & 0.245      & 0.290      & 0.268      & 0.258      & 0.270   \\ \cline{2-8} 
                         & -0.06 & 0.428 & 0.332      & 0.428      & 0.372      & 0.388      & 0.340   \\ \cline{2-8} 
                         & -0.08 & 0.533 & 0.475      & 0.495      & 0.505      & 0.498      & 0.400   \\ \cline{2-8} 
                         & -0.10 & 0.673 & 0.557      & 0.530      & 0.582      & 0.613      & 0.512   \\ \hline
\multirow{6}{*}{FPP}     & 0.00  & 0.048 & 0.037      & 0.272      & 0.020      & 0.157      & 0.043   \\ \cline{2-8} 
                         & -0.02 & 0.120 & 0.142      & 0.455      & 0.122      & 0.333      & 0.160   \\ \cline{2-8} 
                         & -0.04 & 0.277 & 0.288      & 0.747      & 0.153      & 0.522      & 0.285   \\ \cline{2-8} 
                         & -0.06 & 0.447 & 0.380      & 0.892      & 0.253      & 0.697      & 0.455   \\ \cline{2-8} 
                         & -0.08 & 0.638 & 0.508      & 0.965      & 0.320      & 0.850      & 0.610   \\ \cline{2-8} 
                         & -0.10 & 0.787 & 0.675      & 1.000      & 0.433      & 0.930      & 0.773   \\ \hline
\multirow{6}{*}{LPB-PP}  & 0.00  & 0.057 & 0.048      & 0.188      & 0.022      & 0.098      & 0.048   \\ \cline{2-8} 
                         & -0.02 & 0.152 & 0.218      & 0.392      & 0.200      & 0.192      & 0.143   \\ \cline{2-8} 
                         & -0.04 & 0.280 & 0.292      & 0.602      & 0.268      & 0.362      & 0.297   \\ \cline{2-8} 
                         & -0.06 & 0.463 & 0.452      & 0.845      & 0.338      & 0.567      & 0.433   \\ \cline{2-8} 
                         & -0.08 & 0.632 & 0.577      & 0.915      & 0.410      & 0.753      & 0.638   \\ \cline{2-8} 
                         & -0.10 & 0.763 & 0.660      & 0.967      & 0.500      & 0.853      & 0.770   \\ \hline
\multirow{6}{*}{CBB-PP}  & 0.00  & 0.042 & 0.025      & 0.247      & 0.022      & 0.142      & 0.060   \\ \cline{2-8} 
                         & -0.02 & 0.147 & 0.165      & 0.417      & 0.212      & 0.313      & 0.150   \\ \cline{2-8} 
                         & -0.04 & 0.278 & 0.238      & 0.742      & 0.292      & 0.562      & 0.268   \\ \cline{2-8} 
                         & -0.06 & 0.437 & 0.352      & 0.938      & 0.322      & 0.790      & 0.415   \\ \cline{2-8} 
                         & -0.08 & 0.643 & 0.468      & 0.973      & 0.423      & 0.912      & 0.640   \\ \cline{2-8} 
                         & -0.10 & 0.782 & 0.592      & 0.992      & 0.507      & 0.965      & 0.783   \\ \hline
\end{tabular}
\end{table}

\section{Conclusion}
We proposes LPB unit root test to sooth the size distortion of unit root test, in particular the PP test, with MA noises. Via a bootstrap functional central limit theorem, the validity of LPB unit root test is established under general assumptions which allow a large family of non-linear noises. Simulation shows LPB unit root test mitigates the size distortion of the PP test under moving average noises, while preserving its high powers.\\
\\
Hence, LPB unit root test stands out as competitive alternative in testing unit root. Further study will be needed to compare the empirical and the (local) asymptotic efficiency of LPB unit root test and other variants of unit root tests, e.g., the modified test of \textcite{perron1996}.

\section{Appendix}
We first introduce some extra notations. Let $||\cdot||_{p}$ be the $L^{p}$ (induced) norm of vectors (or matrices). Let $||\cdot||$ be $||\cdot||_{2}$. Let $tr(\cdot)$ be the trace of matrices. Let $\bm{1}_{u}$ be a $n$-dimensional column vector with first $\lnul$ entries one and the other entries zero. Let $\bm{1}_{u,v}=\bm{1}_{v}-\bm{1}_{u}$. Define $\hgamma_{V}$, $\hrho_{V}$, and $\hSig_{V}$ analogously as $\hgamma_{\hV}$, $\hrho_{\hV}$, and $\hSig_{\hV}$.
\begin{proof}[Proof of Lemma \ref{lemma:bootstrap FCLT}]
Let $M^{*}$ be a random matrix with rows independently and uniformly selected from the standard basis vectors, e.g., $(1,0,...,0)$, and $\bep^{*}=M^{*}\cbep$. Let  $\hbe=(\he_{1},...,\he_{n})'$,
$\cbe=(\ce_{1},...,\ce_{n})'$,
$\cbvep=(\cvep_{1},...,\cvep_{n})'$,
$\bes=(e^{*}_{1},...,e^{*}_{n})'$, $\bveps=(\vep^{*}_{1},...,\vep^{*}_{n})'$. Let $\bar{\he}_{t}=n^{-1}\sumn {\he}_{t}$, $\hsig_{\he}^{2}=n^{-1}\sumn(\he_{t}-\bar{\he}_{t})^{2}.$ Let $e_{t}^{*}$ and $\vep_{t}^{*}$ be generated with the true autocovariance matrix $\Sig$, as follows:\\
\\
Step 1: $\hbe=\Sig^{-1/2}\cbV$. \\
Step 2: $\ce_{t}=(\he_{t}-\bar{\he}_{t})/\hsig_{\he}$, $\cvep_{t}=(\he_{t}-\bar{\he}_{t})/\hsig_{\hep}$ \\
Step 3: $\be^{*}=M^{*}\cbe$, $\bvep^{*}=M^{*}\cbvep$. \\
\\
Notice $S^{*}(u)=R_{1}^{*}(u)+R_{2}^{*}(u)+R_{3}^{*}(u)$, where
$$R_{1}^{*}(u)={\sigs}^{-1}n^{-1/2}\bm{1}_{u}'\hSig_{\hV}^{1/2}\bes, \  R_{2}^{*}(u)={\sigs}^{-1}n^{-1/2}\bm{1}_{u}'\hSig_{\hV}^{1/2}(\bveps-\bes), \ \text{and} \ R_{3}^{*}(u)={\sigs}^{-1}n^{-1/2}\bm{1}_{u}'\hSig_{\hV}^{1/2}(\beps-\bveps).$$
To show $R_{j}^{*}(u)$ converges, that is, $R_{1}^{*}\Rightarrow W$ in probability, $R_{2}^{*}\Rightarrow 0$ in probability, and $R_{3}^{*}\Rightarrow 0$ in probability, it suffices to prove the in-probability finite-dimensional convergence and the in-probability tightness of $R_{j}^{*}(u)$. With Lemma \ref{lemma:covariance} below, the proof of Theorem 5 of \textcite{mcmurry2010}, and the Cramer-Wold device, it is straightforward to show the finite-dimensional convergence of $R_{1}^{*}(u)$. The finite-dimensional convergence of $R_{2}^{*}(u)$ and $R_{3}^{*}(u)$ can be proven similarly. To establish the in-probability tightness of $R_{j}^{*}(u)$, we apply Theorem 13.5 of \textcite{billingsley1999} (p. 142), and verify its conditions with Lemma \ref{lemma:matrix}, \ref{lemma:sigma}, \ref{lemma:tightness}, and \ref{lemma:fourth term} below. The in-probability convergence of $S^{*}(u)$ follows from Slutsky's Theorem on a metric space (see e.g. \textcite{billingsley1999}, Theorem 3.1, p. 27).
\end{proof}
\begin{proof}[Proof of Theorem \ref{theo:consistency}]
By Lemma \ref{lemma:FCLT}, \ref{lemma:bootstrap FCLT}, and Lemma \ref{lemma:variance}, the conditional asymptotic distributions of $n(\hphis-1)$ and the t-statistic $\ts$ are both standard Phillips-Perron type distributions; see Theorem 3.1 of \textcite{phillips1987}. So do the unconditional asymptotic distributions of $n(\hphi-1)$ and the t-statistic $t$. Further, by Lemma \ref{lemma:sigma} and \ref{lemma:variance}, the parameters in the conditional asymptotic distributions converge to the parameters in the unconditional asymptotic distributions. The theorem follows. 
\end{proof}
\begin{Lemma}\label{lemma:phi}
Under Assumption \ref{assump:weak},
$$\hphi-\phi=\begin{cases}
O_{p}(n^{-1/2}),      & \text{if } \phi < 1, \\
O_{p}(n^{-1}),  & \text{if } \phi=1.
\end{cases}$$ 
\end{Lemma}
\begin{proof}[Proof of Lemma \ref{lemma:phi}]
This result follows straightforwardly from Theorem 3 of \textcite{wu2005}.
\end{proof}
\begin{Lemma}\label{lemma:matrix}
Under Assumption \ref{assump:weak} and Condition \ref{cond:bandwidth},
$$||\hSig_{\hV}-\Sig||=O_{p}(n^{-1/4}),\text{ and }||\hSig_{\hV}^{-1}-\Sig^{-1}||=O_{p}(n^{-1/4}).$$
\end{Lemma}
\begin{proof}[Proof of Lemma \ref{lemma:matrix}]
By the proof of Theorem 2 and 3 of \textcite{mcmurry2010}, it suffices to prove $||\tSig_{\hV}-\Sig||=O_{p}(r_{n}).$ By Theorem 1 of \textcite{mcmurry2010},  $||\tSig_{V}-\Sig||=O_{p}(r_{n}),$ so it suffices to prove $|\tSig_{\hV}-\tSig_{V}||=O_{p}(r_{n}),$ where $\tSig_{\hV}$ is defined in Algorithm \ref{algor:matrix} and $r_{n}$ in Condition \ref{cond:bandwidth}. By H\"older's inequality and the symmetry of $\hSig_{\hV}$ and $\hSig_{V}$, 
\begin{align*}
&\mathrel{\phantom{=}}||\tSig_{\hV}-\tSig_{V}||\leq||\tSig_{\hV}-\tSig_{V}||_{1}\leq 2\sum_{h=0}^{\lfloor c_{\kappa}l \rfloor}|\hgamma_{\hV}(h)-\hgamma_{V}(h)|\\
&\leq 2(l'+1)(|\hphi-\phi|(C_{1}+C_{2})+(\hphi-\phi)^{2}C_{3}), 
\end{align*}
where $l'=\lfloor c_{\kappa}l \rfloor$, and
\begin{align*}
C_{1}&=\sup_{0\leq h\leq l'}|n^{-1}\sum_{t=h+1}^{n}Y_{t-1}V_{t-h}|,\\
C_{2}&=\sup_{0\leq h\leq l'}|n^{-1}\sum_{t=h+1}^{n}Y_{t-h-1}V_{t}|,\\
C_{3}&=\sup_{0\leq h\leq l'}|n^{-1}\sum_{t=h+1}^{n}Y_{t-h-1}Y_{t-1}|.
\end{align*}
When $\phi< 1$, by Theorem 1 of \textcite{hannan1974}, $$C_{1}=\sup_{0\leq h\leq l'}|n^{-1}\sum_{t=h+1}^{n}Y_{t-1}(Y_{t-h}-\phi Y_{t-h-1})|=\Op.$$
When $\phi=1$, 
\begin{align*}
C_{1}&=\sup_{0\leq h\leq l'} |n^{-1}\sum_{t=h+1}^{n}Y_{t-h-1}V_{t-h}+n^{-1}\sum_{t=h+1}^{n}\sum_{k=1}^{h}V_{t-k}V_{t-h}|\\
&=\sup_{0\leq h\leq l'} |(2n)^{-1}(\sum_{t=h+1}^{n}(Y_{t-h-1}+V_{t-h})^{2}-Y_{t-h-1}^{2}-V_{t-h}^{2})+n^{-1}\sum_{t=h+1}^{n}\sum_{k=1}^{h}V_{t-k}V_{t-h}|\\
&=\sup_{0\leq h\leq l'} |(2n)^{-1}Y_{n-h}^{2}-(2n)^{-1}\sum_{t=h+1}^{n}V_{t-h}^{2}+n^{-1}\sum_{t=h+1}^{n}\sum_{k=1}^{h}V_{t-k}V_{t-h}|\\
&=\Op+\Op+O_{p}(l+l^{3}n^{-1})=O_{p}(l+l^{3}n^{-1}),
\end{align*}
since $$\sup_{0\leq h\leq l'} n^{-1}Y_{n-h}^{2}\leq (\sup_{0 \leq u\leq 1}n^{-1/2}Y_{\lfloor nu \rfloor})^{2}=\Op,\ \sup_{0\leq h\leq l'} n^{-1}\sum_{t=h+1}^{n}V_{t-h}^{2}=n^{-1}\sumn V_{t}^{2}=\Op,$$
and
\begin{align*}
&\mathrel{\phantom{=}}\sup_{0\leq h\leq l'}|n^{-1}\sum_{t=h+1}^{n}\sum_{k=1}^{h}V_{t-k}V_{t-h}|\leq(l'+1)\sup_{0\leq h\leq l'}\sup_{0<k\leq h}|n^{-1}\sum_{t=h+1}^{n}V_{t-k}V_{t-h}|\\
&\leq (l'+1)(\sup_{j\geq 0, k>0, j+k\leq l'}|n^{-1}\sum_{t=n-k+1}^{n}V_{t}V_{t-j}|+\sup_{0\leq j\leq l'}|n^{-1}\sum_{t=j+1}^{n}V_{t}V_{t-j}|)\\
&\leq (l'+1)^{2}(2n)^{-1} \sup_{n-l<t\leq n, 0\leq j<l}(V_{t}^{2}+V_{t-j}^{2})+O_{p}(l)=O_{p}(l^{3}n^{-1}+l).  
\end{align*}
Similarly, it can be shown that when $\phi< 1$, $C_{2}=\Op$, and $C_{3}=\Op,$ and when $\phi=1$, $C_{2}=O_{p}(l+l^{3}n^{-1})$, and $C_{3}=O_{p}(n+l^{2}+l^{4}n^{-1})$. By Lemma \ref{lemma:phi},
$$||\tSig_{\hV}-\tSig_{V}||=\begin{cases}
O_{p}(ln^{-1/2}),      & \text{if } \phi < 1, \\
O_{p}(l^{2}n^{-1}+l^{4}n^{-2}),  & \text{if } \phi=1. \\
\end{cases}$$ 
\end{proof}
\begin{Lemma}\label{lemma:sigma} 
Under Assumption \ref{assump:weak} and Condition \ref{cond:bandwidth},
$$\Vars (n^{-1/2}\sumn \Vs_{t})-Var(n^{-1/2}\sumn V_{t})=\op.$$
\end{Lemma}
\begin{proof}[Proof of Lemma \ref{lemma:sigma}]
The result follows from Lemma \ref{lemma:matrix}, and Lemma 3 and 4 of \textcite{mcmurry2010}. 
\end{proof}
\begin{Lemma}\label{lemma:covariance}
Suppose Assumption \ref{assump:weak} and Condition \ref{cond:bandwidth} hold. For all $0\leq u\leq v\leq 1$, $$\Covs(R^{*}_{1}(u),R^{*}_{1}(v))=u+\op.$$
\end{Lemma}
\begin{proof}[Proof of Lemma \ref{lemma:covariance}]
By Lemma \ref{lemma:sigma}, for all $0\leq w \leq 1$,
$$n^{-1}\bm{1}_{w}'\hSig_{\hV}\bm{1}_{w}=n^{-1}\lnwl \Vars(\lnwl^{-1/2}\sum_{t=1}^{\lnwl}\Vs_{t})=\sig^{2}w+\op.$$
Hence,
$$\Covs(R^{*}_{1}(u),R^{*}_{1}(v))={\sigs}^{-2}n^{-1}\bm{1}_{u}'\hSig_{\hV}\bm{1}_{v}={\sigs}^{-2}(2n)^{-1}(\bm{1}_{v}'\hSig_{\hV}\bm{1}_{v}+\bm{1}_{u}'\hSig_{\hV}\bm{1}_{u}-\bm{1}_{v-u}'\hSig_{\hV}\bm{1}_{v-u})=u+\op.$$
\end{proof}
\begin{Lemma}\label{lemma:tightness}
Suppose $A$ is a $n\times n$-dimensional symmetric positive semi-definite matrix, $A=A^{1/2}{A^{1/2}}'$, and $A^{1/2}=\{a_{ij}\}_{i,j=1}^{n}$. Suppose $\xi_{1}^{*},...,\xi^{*}_{n}$ are $\Ps$-i.i.d random variables with  $\Es\xi^{*}_{t}=0$ and  $\Es({\xi_{j}^{*}}^{2})={s^{*}}^{2}$. Let $R^{*}(u)={\sigs}^{-1}n^{-1/2}\bm{1}_{u}'A^{1/2}\bxis$. Then, for all $0\leq u \leq v \leq w\leq 1$, 
$$\Es((R^{*}(v)-R^{*}(u))^{2}(R^{*}(w)-R^{*}(v))^{2})\leq4{\sigs}^{-4}\Es({\xi_{j}^{*}}^{4})||A||^{2}(w-u)^{2}.$$
\end{Lemma}
\begin{proof}[Proof of Lemma \ref{lemma:tightness}]
\begin{align}
\label{eqn:tightness 1}
&\mathrel{\phantom{=}}\Es((R_{1}^{*}(v)-R_{1}^{*}(u))^{2}(R_{1}^{*}(w)-R_{1}^{*}(v))^{2})={\sigs}^{-4}n^{-2}\Es ((\bm{1}_{u,v}'A^{1/2}\bxis)^{2}(\bm{1}_{u,v}'A^{1/2}\bxis)^{2})\nonumber \\
&={\sigs}^{-4}n^{-2}E^{*}((\sum_{j=1}^{n}\sum_{i=\lnul+1}^{\lnvl}a_{ij}\xi^{*}_{j})^{2}(\sum_{j=1}^{n}\sum_{i=\lnvl+1}^{\lnwl}a_{ij}\xi^{*}_{j})^{2})
={\sigs}^{-4}(B_{1}+B_{2}+B_{3}),
\end{align}
where
\begin{align*}
%=\Es({\xi_{j}^{*}}^{4}-3)n^{-2}\sum_{j=1}^{n}\sum_{i_{1}=\lnul+1}^{\lnvl}\sum_{i_{2}=\lnul+1}^{\lnvl}\sum_{i_{3}=\lnvl+1}^{\lnwl}\sum_{i_{4}=\lnvl+1}^{\lnwl}a_{i_{1}j}a_{i_{2}j}a_{i_{3}j}a_{i_{4}j}\\
&\mathrel{\phantom{=}}B_{1}={s^{*}}^{4}n^{-2}(\sum_{j=1}^{n}(\sum_{i=\lnul+1}^{\lnvl}a_{ij})^{2})(\sum_{j=1}^{n}(\sum_{i=\lnvl+1}^{\lnwl}a_{ij})^{2})\\
&\mathrel{\phantom{=}}B_{2}=2{s^{*}}^{4}n^{-2}(\sum_{j=1}^{n}(\sum_{i=\lnul+1}^{\lnvl}a_{ij})(\sum_{i=\lnvl+1}^{\lnwl}a_{ij}))^{2}\\
&\mathrel{\phantom{=}}B_{3}=\Es({\xi_{j}^{*}}^{4}-3{s^{*}}^{4})n^{-2}\sum_{j=1}^{n}(\sum_{i=\lnul+1}^{\lnvl}a_{ij})^{2}(\sum_{i=\lnvl+1}^{\lnwl}a_{ij})^{2}
\end{align*}
Notice $B_{1}={s^{*}}^{4}n^{-2}\bm{1}_{u,v}'A\bm{1}_{u,v}\bm{1}_{v,w}'A\bm{1}_{v,w}$, and $B_{2}=2{s^{*}}^{4}n^{-2}(\bm{1}_{u,v}'A\bm{1}_{v,w})^{2}$. Since 
$$\bm{1}_{u,v}'A\bm{1}_{v,w}=2^{-1}(\bm{1}_{u,w}'A\bm{1}_{u,w}-\bm{1}_{u,v}'A\bm{1}_{u,v}-\bm{1}_{v,w}'A\bm{1}_{v,w}),$$
and for $0\leq r\leq s\leq 1$,
$$\bm{1}_{r,s}'A\bm{1}_{r,s}\leq ||A||(\lnrl-\lnsl),$$
we have 
\begin{align*}
&B_{1}\leq {s^{*}}^{4}||A||^{2} ((\lnvl-\lnul)/n)((\lnwl-\lnvl)/n)\leq 4{s^{*}}^{4}||A||^{2}(w-u)^{2},\\
&B_{2}\leq 2{s^{*}}^{4}||A||^{2} ((\lnwl-\lnul)/n)^{2}\leq 8{s^{*}}^{4}||A||^{2}(w-u)^{2},\\
&B_{3}\leq \Es({\xi_{j}^{*}}^{4}-3{s^{*}}^{4}){s^{*}}^{-4}B_{1}\leq4\Es({\xi_{j}^{*}}^{4}-3{s^{*}}^{4})||A||^{2}(w-u)^{2}.
\end{align*}
The lemma follows from \eqref{eqn:tightness 1}.
\end{proof}
\begin{Lemma}\label{lemma:fourth term}
Under Assumption \ref{assump:weak} and Condition \ref{cond:bandwidth},
$$(i)\ \Es({e_{j}^{*}}^{4})=\Op, \ (ii)\ \Es(({\vep_{j}^{*}}-{e_{j}^{*}})^{4})=\Op, \ \text{and} \  (iii)\ \Es(({\ep_{j}^{*}}-{\vep_{j}^{*}})^{4})=\Op.$$
\end{Lemma}
\begin{proof}[Proof of Lemma \ref{lemma:fourth term}]
The proof applies Lemma \ref{lemma:phi} through out. Notice that
\begin{equation} \label{eqn:fourth term 1}
n^{-1}\sumn\he_{t}^{2}=n^{-1}\cbV'\Sig^{-1}\cbV\geq ||\Sig||^{-1}n^{-1}\sumn \cV_{t}^{2}=||\Sig||^{-1}\gamma_{0}+\op.
\end{equation}
By Chebyshev's Inequality,
\begin{equation} \label{eqn:fourth term 2}
{\bar{\he}_{t}}^{2}\leq 2(n^{-1}\bm{1}'\Sig^{-1/2}(\cbV-\bV))^{2}+2(n^{-1}\bm{1}'\Sig^{-1/2}\bV)^{2})\leq
2||\Sig^{-1}||n^{-1}
||\cbV-\bV||^{2}+\op=\op.
\end{equation}
By \eqref{eqn:fourth term 1} and \eqref{eqn:fourth term 2}, $\hsig_{\he}^{-2}=\Op$. Similarly, $\hsig_{\hep}^{-2}=\Op.$ Further,
$$n^{-1}\sumn \he_{t}^{4}=n^{-1}||\Sig^{-1/2}\cbV||_{4}^{4}\leq 8n^{-1}(||\Sig^{-1/2}(\cbV-\bV)||_{4}^{4}+||\Sig^{-1/2}\bV||_{4}^{4})=\Op,$$
since by Lemma 5 of \textcite{mcmurry2010},
$$n^{-1}||\Sig^{-1/2}\bV||_{4}^{4}=\Op,$$ 
and 
\begin{align*}
n^{-1}||\Sig^{-1/2}(\cbV-\bV)||_{4}^{4} &\leq n^{-1}||\Sig^{-1/2}(\cbV-\bV)||^{4}\leq 8||\Sig^{-1}||^{2}n^{-1}(||\hbV-\bV||^{4}+||\bar{\hbV}||^{4})\\
&\leq 64 ||\Sig^{-1}||^{2}((\hphi-\phi)^{4}n^{-1}\sumn Y_{t-1}^{4}+\bar{V}_{t}^{4}+(\hphi-\phi)^{4}\bar{Y}_{t}^{4})=\Op.
\end{align*}
For (i) and (ii), therefore,
\begin{align*}
&\Es({e_{j}^{*}}^{4})=\hsig_{\he}^{-4}n^{-1}\sumn(\he_{t}-\bar{\he}_{t})^{4}=\Op,\\
&\Es(({\vep_{j}^{*}}-{e_{j}^{*}})^{4})=(\hsig_{\hep}^{-1}-\hsig_{\he}^{-1})^{4}n^{-1}\sumn(\he_{t}-\bar{\he}_{t})^{4}=\Op.
\end{align*}
For (iii), by Lemma \ref{lemma:matrix},
\begin{align*}
&\mathrel{\phantom{=}}\Es(({\ep_{j}^{*}}-{\vep_{j}^{*}})^{4})=\hsig_{\hep}^{-4}n^{-1}\sumn(\hep_{t}-\bar{\hep}_{t}-(\he_{t}-\bar{\he}_{t}))^{4}\leq \hsig_{\hep}^{-4}n^{-1}(\sumn(\hep_{t}-\bar{\hep}_{t}-(\he_{t}-\bar{\he}_{t}))^{2})^{2}\\
&=\hsig_{\hep}^{-4}n^{-1}||(I-n^{-1}\bm{1}\bm{1}')(\hSig_{\hV}^{-1/2}-\Sig^{-1/2})\cbV||^{4}\leq \hsig_{\hep}^{-4}n||\hSig_{\hV}^{-1/2}-\Sig^{-1/2}||^{4}(n^{-1}\sumn\cV_{t}^{2})^{2}=\Op.
\end{align*}
\end{proof}
\begin{Lemma}\label{lemma:variance}
Suppose Assumption \ref{assump:weak} and Condition \ref{cond:bandwidth} hold. Then  $$\Vars(n^{-1}\sumn{\Vs_{t}}^{2})=o_{p}(1),\ \text{and} \ \Es(n^{-1}\sumn{\Vs_{t}}^{2})=\gamma_{0}+\op.$$
\end{Lemma}
\begin{proof}[Proof of Lemma \ref{lemma:variance}]
Notice that $\sumn {\Vs_{t}}^{2}={\beps}'\hSig\beps$. By Lemma \ref{lemma:matrix} and \ref{lemma:fourth term} above and \textcite{seber2012} (Theorem 1.5 and 1.6, pp. 9-10),
$$\Es(n^{-1}\sumn{\Vs_{t}}^{2}-\gamma_{0})=n^{-1}tr(\hSig-\Sig)\leq ||\hSig-\Sig||=\op,$$
$$\Vars(n^{-1}\sumn {\Vs_{t}}^{2})\leq n^{-2}E({\ep_{t}^{*}}^{4})tr(\hSig^{2})\leq n^{-1}E({\ep_{t}^{*}}^{4})||\hSig||^{2}=\op.$$
\end{proof}
\begin{proof}[Proof of Lemma \ref{lemma:bandwidth}]
We now prove 
\begin{equation}\label{eqn:bandwidth1}
\sum_{h=\hl+1}^{\infty}|\gamma(h)|=O_{p}(n^{-1/4}).
\end{equation}
Assume at this stage $\gamma(h)\neq0$ for infinitely many $h$. Let $g_{h}=|\gamma(h)|$, $G_{h}=\sum_{k=h+1}^{\infty}g_{k}$, $G^{-1}(x)=\min \{h\geq 0: G_{h}\leq x\}$, and $a=G^{-1}(n^{-1/4})$. Then
\begin{align}\label{eqn:bandwidth2}
P(\sum_{h=\hl+1}^{\infty}|\gamma(h)|>n^{-1/4})&=P(\hl<a)=1-P(\forall l=1,...,a-1, \sup_{1\leq k\leq K_{n}}|\hrho_{\hV}(l+k)|\geq c(\log n/n)^{1/2})\nonumber \\
&=1-(D_{1}-D_{2}-D_{3}),
\end{align}
where
\begin{align*}
&D_{1}=P(\forall l=1,...,a-1, \sup_{1\leq k\leq K_{n}}|\rho(l+k)|\geq 3c(\log n/n)^{1/2}),\\    
&D_{2}=P(\exists l=1,...,a-1, \sup_{1\leq k\leq K_{n}}|\hrho_{V}(l+k)-\rho(l+k)|> c(\log n/n)^{1/2})),\\ 
&D_{3}=P(\exists l=1,...,a-1, \sup_{1\leq k\leq K_{n}}|\hrho_{\hV}(l+k)-\hrho_{V}(l+k)|> c(\log n/n)^{1/2})).
\end{align*}
By the proof of Lemma \ref{lemma:matrix} and Theorem 1 of \textcite{xiao2011}, 
\begin{equation} \label{eqn:bandwidth3}
D_{2}=o(1) \ \text{and} \ D_{3}=o(1).
\end{equation}
Now we show $D_{1}=1+o(1)$. Let $f(l)=\sup_{1\leq k}g_{l+k}$ and $f_{n}(l)=\sup_{1\leq k\leq K_{n}}g_{l+k}$. For some $0<D<1$, 
$$\inf_{1\leq l<a}f_{n}(l)\geq \inf_{1 \leq l<a}f(l)-\sup_{1 \leq l<a}|f_{n}(l)-f(l)|\geq \sup_{k\geq a}g_{k}-\sup_{k\geq a+K_{n}}g_{k}\geq D \sup_{k\geq a}g_{k}\geq D g_{a}.$$
Hence, for some $C>0$,
\begin{align}\label{eqn:bandwidth4}
D_{1}&\geq P(\inf_{1\leq l<a}f_{n}(l)\geq C(\log n/n)^{1/2})\geq P( g_{a}\geq (C/D)(\log n/n)^{1/2})\nonumber \\
&=P( g_{G^{-1}(n^{-1/4})}\geq (C/D)(\log n/n)^{1/2})=1+o(1),
\end{align}
where the last equation results from Lemma \ref{lemma:inverse covariance} below. A combination of \eqref{eqn:bandwidth2}, \eqref{eqn:bandwidth3}, and \eqref{eqn:bandwidth4} gives \eqref{eqn:bandwidth1} when $\gamma(h)\neq 0$ for infinitely many $h$. When $\gamma(h)\neq 0$ only for finitely many $h$, \eqref{eqn:bandwidth1} follows analogously. It can be similarly derived that 
$$\hl n^{-1/2}=O_{p}(n^{-1/4}).$$

%To show $$P(\phi_{a}\geq C(\log n/n)^{1/2})\rightarrow 1,$$
%it suffices to show $g_{a}>n^{(\beta/2-1)/2}$, or $g_{G^{-1}(x)}$ Assume 
%Then 
%$$n^{1/4}\phi_{a}\geq n^{1/4}\phi_{b}\geq b\phi_{b}\geq \Phi_{b}=n^{-1/4}.$$
%So if \eqref{eqn:1} holds, $\phi_{a}\geq n^{-1/2}$, and it can be similarly shown that $\phi_{a}\geq (\log n/n)^{1/2}$, $P(\hl\geq a)\rightarrow 1$, and $P(G(\hl)>n^{-1/4})\rightarrow 0$.
\end{proof}
\begin{Lemma}\label{lemma:inverse covariance}
Suppose Assumption \ref{assump:weak} and Assumption \ref{assump:strong} hold. Suppose $\gamma(h)\neq 0$ for infinitely many $h$. Let $g_{h}=|\gamma(h)|$, $G_{h}=\sum_{k=h+1}^{\infty}g_{k}$, $G^{-1}(x)=\min \{h\geq 0: G_{h}\leq x\}$. Then for a small enough positive number $x$, $$g_{G^{-1}(x)}>(\alpha/4)x^{\beta/(\beta-1)}.$$
\end{Lemma}
\begin{proof}[Proof of Lemma \ref{lemma:inverse covariance}]
If $g_{h}=o(h^{-\beta})$, then $G_{h}=o(h^{1-\beta})$, and then for a small enough positive number $x$,
\begin{equation}\label{eqn:inverse covariance1}
G^{-1}(x)<x^{1/(1-\beta)}. 
\end{equation}
By Assumption \ref{assump:strong}, for large enough $h$, $h^{\alpha}G_{h}$ is non-increasing. It follows straightforwardly that for large enough $h$,
\begin{equation}\label{eqn:inverse covariance2}
hg_{h}\geq (\alpha/2)G_{h}.
\end{equation}
Hence, by \eqref{eqn:inverse covariance1} and \eqref{eqn:inverse covariance2}, for a small enough positive number $x$, 
\begin{align*}
g_{G^{-1}(x)}=\frac{G^{-1}(x)g_{G^{-1}(x)}}{(\alpha/2)G_{G^{-1}(x)}}\cdot\frac{(\alpha/2)G_{G^{-1}(x)}}{G^{-1}(x)}>(\alpha/4)x^{\beta/(\beta-1)}.
\end{align*}
\end{proof}
\printbibliography
\end{document}